\documentclass[a4paper,11pt]{article}

\usepackage[T1]{fontenc}
\usepackage[utf8]{inputenc}
\usepackage{lmodern}

\usepackage{geometry}
\usepackage{setspace}
\usepackage{amsmath, amssymb, amsthm, amsfonts, bm}
\usepackage{mathtools}

\usepackage{graphicx}
\usepackage{float}
\usepackage{tikz}
\usetikzlibrary{positioning}

\usepackage{algorithm}
\usepackage{algpseudocode}

\usepackage{booktabs}
\usepackage{longtable}
\usepackage{rotating}

\usepackage{enumitem}

\usepackage{xcolor}
\usepackage{soul}

\usepackage{pdflscape}

\usepackage{hyperref}
\hypersetup{
  colorlinks=true,
  linkcolor=blue!60!black,
  citecolor=blue!60!black,
  urlcolor=blue!60!black,
  pdfauthor={Nir Shvalb and Oded Medina},
  pdftitle={Hypo-paradoxical Linkages: Linkages That Should Move—But Don’t}
}
\usepackage[authoryear,longnamesfirst]{natbib}

\newtheorem{theorem}{Theorem}[section]
\newtheorem{lemma}[theorem]{Lemma}

\newtheorem{proposition}[theorem]{Proposition}
\newtheorem{observation}[theorem]{Observation}

\theoremstyle{definition}
\newtheorem{definition}[theorem]{Definition}

\newtheorem{remark}[theorem]{Remark}
\newtheorem{notation}[theorem]{Notation}

\newcommand{\C}{\mathcal{C}}
\newcommand{\LL}{\mathcal{L}}
\newcommand{\II}{\mathcal{I}}
\newcommand{\IM}{\bar{M}}

\newcommand{\R}{\mathbb{R}}

\newcommand{\bq}{\bm{q}}
\newcommand{\bomega}{\bm{\omega}}
\newcommand{\bxi}{\bm{\xi}}

\title{Hypo-paradoxical Linkages: Linkages That Should Move—But Don’t}

\author{
  Nir Shvalb\thanks{Corresponding author: \texttt{nirsh@ariel.ac.il}. ORCID: 0000-0001-8246-3727} \\
  \small Department of Mechanical Engineering, Ariel University, Ariel 40700, Israel
  \and
  Oded Medina\thanks{ORCID: 0000-0001-6590-0059} \\
  \small Department of Mechanical Engineering, Ariel University, Ariel 40700, Israel
}
\date{} % empty date

\begin{document}
\maketitle
\begin{abstract}
While paradoxical linkages famously violate the Chebyshev–Gr\"ubler–Kutzbach criterion by exhibiting unexpected mobility, we identify an opposing phenomenon: a class of linkages that appear mobile according to the same criterion, yet are in fact rigid.  We refer to these as \emph{hypo-paradoxical linkages}, and proceed to analyze and illustrate their behavior. We use the same tools to further explain the unexpected positive mobility of Bennett mechanism.
\end{abstract}


%\item Introducing hypo-paradoxical linkages: predicted to have mobility, yet are rigid.
%\item Presents a rich family of spatial linkages that exhibit hypo-paradoxical behavior.
%\item Introducing the immobility margin, linking workspace size to fabrication error.
%\item Providing a simple and intuitive geometric explanation of the Bennett paradox.

% Keywords


%%%%%%%%%%%%%%%%%%%%%%%
\section{Introduction}
 %%%%%%%%%%%%%%%%%%%%%%%
%
Classical mobility criteria, such as the Chebyshev–Gr\"ubler–Kutzbach formula (1904), offer a quick estimate for the number of degrees-of-freedom in linkage systems based on link and joint counts. However, several historical linkages defy these predictions, exhibiting unexpected motion despite appearing overconstrained. Such linkages are known as \textit{paradoxical linkages}. In addition to Gr\"ubler's formula, rigidity criteria such as \emph{Laman’s theorem} provide a somewhat more complex condition for the rigidity of planar linkages. The theorem states that a graph with \( n \) vertices is \emph{minimally rigid} in the plane if it has exactly \( 2n - 3 \) edges, and no subset of \( k \) vertices spans more than \( 2k - 3 \) edges.  However, Laman’s criterion also breaks down in the presence of non generic linkages necessitating deeper analysis of configuration spaces and infinitesimal rigidity (cf. \cite{connelly2005generic,kapovich2002moduli}).

For example, the Sarrus linkage \cite{sarrus1853} is an early paradoxical mechanism, achieving straight-line motion via a spatial six-bar chain of revolute joints despite being over-constrained. Bennett \cite{bennett1903} introduced a four-bar spatial linkage that defies the Gr\"ubler formula through special equalities in its link lengths and hinge angles. Shortly thereafter, Bricard classified the only possible single-loop 6R linkages, describing several families of mobile octahedral mechanisms that remain flexible under some symmetric geometrical conditions \cite{bricard1897}. Delassus \cite{delassus1922} expanded these insights—examining exotic cases with helical joints and proved that aside from Bennett’s case, no other generic 4R loop can have finite mobility (e.g., \cite{leea2011synthesize}). Myard \cite{myard1931} combined two Bennett linkages into a single five-bar closed chain, a design later recognized as a special case of Goldberg’s broader family of paradoxical linkages. Goldberg further generalized this approach by fusing Bennett loops to obtain new one- degrees-of-freedom  five-bar and six-bar linkages \cite{goldberg1943}. Decades later, Waldron \cite{waldron1967,waldron1979} catalogued additional over-constrained architectures and developed algebraic criteria to systematically derive them (see a review by \cite{baker1979}). More recently, Wohlhart \cite{wohlhart1987,wohlhart1991} discovered novel 6R linkages (for example, by merging two Goldberg 5R chains) and further solidified the modern understanding of paradoxical mobility.
\vspace{5mm}
While investigating classical paradoxical linkages, we unexpectedly discovered an inverse phenomenon. Specifically, we identified mechanisms that, despite appearing mobile under standard degrees-of-freedom analysis, are in fact immobile due to subtle geometric constraints. We term these systems \textit{hypo-paradoxical linkages}, a class of mechanisms that, to the best of our knowledge, has not been previously recognized in the literature.
In this work, we introduce and formalize this concept, presenting a systematic study of such linkages and uncovering the latent geometric constraints responsible for their unexpected rigidity. Using the same analytical tools and geometric reasoning, we also provide an intuitive explanation for the Bennett mechanism’s single degree of freedom. 

\begin{definition}
A \textit{hypo-paradoxical linkage} is a closed kinematic chain \( \mathcal{L} \) for which the classical mobility formula (e.g., Gr\"ubler–Kutzbach) predicts a positive number of degrees of freedom, \( M(\mathcal{L}) > 0 \), yet its actual configuration space has dimension zero, i.e., \( \dim \mathcal{C}(\mathcal{L}) = 0 \). In other words, all admissible configurations (up to rigid-body motions) are isolated.
\end{definition}
This paper formalizes the notion of hypo-paradoxical linkages and exhibits a surprisingly large family of these linkages. The examples provided reveal the limitations of existing mobility formulas and suggest the existence of entire classes of hidden constraints that require systematic investigation.
\section{Closed kinematic chains}
Recall Chebyshev–Gr\"ubler–Kutzbach criterion for spatial linkages
$M = 6(n - 1 - j) + \sum_{i=1}^j f_i$, where \( n \) is the number of links (including the ground), and \( j \) is the number of joints, and \( f_i \) is the number of degrees-of-freedom at the \( i \)-th joint. For an $n$-link closed chain, the formula reads
\begin{equation}
M = 6(n - 1 - n) + n = n - 6.
\label{eq:Grubler}
\end{equation}

Thus, the expected mobility for \( n = 7 \), for example, is predicted to be $1$ (e.g. Bricard-type spatial linkages \cite{bricard1897octaedre}), for $n=8$, Equation \ref{eq:Grubler} implies $M=2$ etc.

To facilitate the forthcoming analysis, we briefly recall the standard screw-theoretic formulation of revolute joints and their associated exponential mappings. 

Each revolute joint is characterized by a screw axis \( \bxi_i \in \mathbb{R}^6 \), expressed in the space (fixed) frame. The closure condition of the kinematic chain is verified by evaluating the product of exponentials:
\begin{equation}
\prod_{i=1}^{n}  e^{\hat{\bxi}_i \theta_i} = 
\begin{bmatrix}
I & 0 \\
0 & 1
\end{bmatrix}
\label{eq:closure}
\end{equation}

The screw axis \( \bxi_i \) corresponding to joint \( i \) is given by:
\begin{equation}
\bxi_i = 
\begin{bmatrix}
\bomega_i \\
\bq_i \times \bomega_i
\end{bmatrix}
\label{eq:screw}
\end{equation}
where \( \bomega_i \in \mathbb{R}^3 \) is the unit rotation axis, and \( \bq_i \in \mathbb{R}^3 \) is any point on the axis.

The twist \( \bxi_i \) can be expressed in matrix form using the \textit{hat operator}:
\begin{equation}
\hat{\bxi}_i =
\begin{bmatrix}
[\bomega_i]_\times & \bq_i \times \bomega_i \\
\bm{0} ^\top & 0
\end{bmatrix}
\label{eq:xi}
\end{equation}
where \( [\bomega_i]_\times \in \mathbb{R}^{3 \times 3} \) denotes the skew-symmetric matrix:
\begin{equation}
[\bomega_i]_\times = 
\begin{bmatrix}
0 & -\omega_{i3} & \omega_{i2} \\
\omega_{i3} & 0 & -\omega_{i1} \\
-\omega_{i2} & \omega_{i1} & 0
\end{bmatrix}
\label{eq:omega}
\end{equation}

Note that Equation \ref{eq:closure} (and its derived differential equation) does not depend on the realization of the physical links connecting them and neither by the choice of the coordinate systems. In other words:
\begin{observation} 
 The mobility (modulo rigid body motions) of a closed linkage depends only on its screw axes \(\{\xi_1, \dots, \xi_n\}\). That is, it is independent of the physical lengths, shapes, or masses of the links connecting the joints.
\label{lemma:axis}
\end{observation}
This is further clarified  in the proof of Theorem \ref{thm}.

%%%%%%%%%%%%%%%%%%%%%%%
 \section{Regular $n$-gon linkages }
%%%%%%%%%%%%%%%%%%%%%%%

Consider, now, a spatial closed-chain linkage conssiting of $n$ revolute joints connected by unit-length links (as in Figure \ref{f:regular_ngon} for $n=7$).  Each joint is associated with a revolute screw axis.   The first joint axis is represented by a unit vector \( \bomega_1 \in \mathbb{R}^3 \), and passes through a point \( \bq_1 \in \mathbb{R}^3 \). Each subsequent joint axis \( \bomega_i \) is obtained by rotating \( \bomega_1 \) by an angle \( \alpha_i = \frac{2\pi(i - 1)}{n}   \)  about the global \( z \)-axis . The corresponding point \( \bq_i \) on the \( i \)-th axis is also obtained by rotating \( \bq_1 \) similarly, i.e.  $\bomega_i = R_z(\alpha_{i-1})\bomega_1$, and $\bq_i = R_z(\alpha_{i-1}) \bq_1
$ for \( i = 1, 2, \dots, n \).   

Note that there are two degenerated cases which we shall exclude for reasons that will be explained in due course:
   \begin{enumerate}[label=(\roman*),leftmargin=*]
    \item The case where $\bomega_1 = \hat{\bm{k}}$ i.e., axes are pointing up.
    \item The case where $\{ \xi_i\}_{i=1}^n$ form a pencil of lines passing through a common point on the axis of symmetry (e.g., all lines intersect at the midpoint of the linkage).
\end{enumerate}
\begin{notation}
We call such linkages \emph{(generic) regular $n$-gon linkages}.    
\end{notation}
\begin{figure}[htb]
\centering
\scalebox{0.5}{\includegraphics{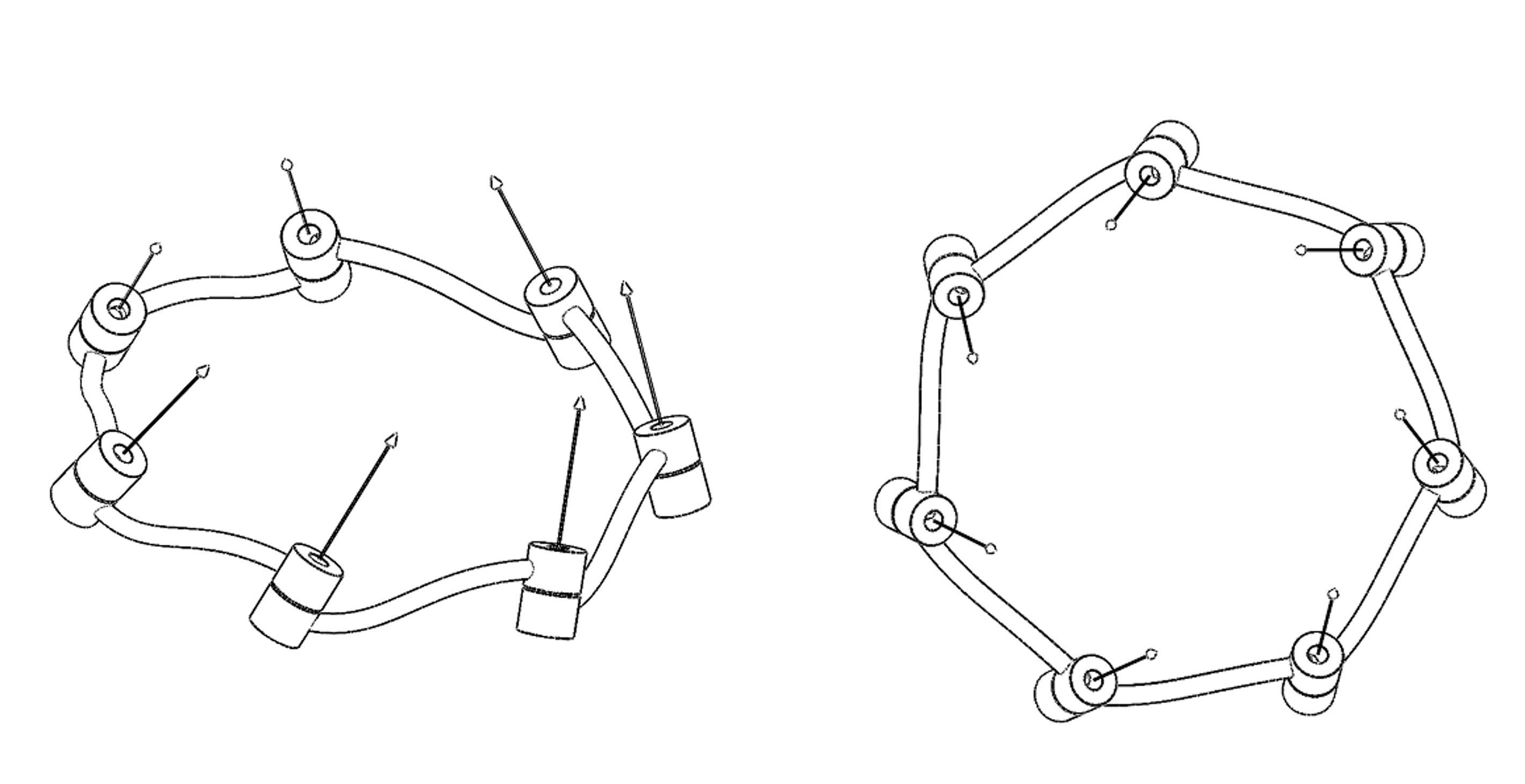}}
\caption{Two projections of a regular $7$-gon linkage which is a hypo-paradoxical linkage.}
\label{f:regular_ngon}
\end{figure}
%
%%%%%%%%%%%%%%%%%%%%%%%
 \subsection{Mobility analysis}
%%%%%%%%%%%%%%%%%%%%%%%

While the Jacobian matrix of a linkage describes the space of instantaneous velocities near a given configuration, the vanishing of all such velocities—i.e., a zero tangent space—does not necessarily imply immobility. This is because the Jacobian captures only first-order (infinitesimal) behavior. A linkage may still admit higher-order infinitesimal motions that escape detection at the linear level. Therefore, establishing rigidity requires analyzing whether any continuous path exists in configuration space, beyond simply checking whether the initial velocity vanishes.
Nevertheless, deriving a full parametric polynomial representation of the configuration space is typically out of reach due to the algebraic complexity of the constraints involved (furthermore, such an analysis would typically be stated for a specific linkage).

\indent Hence we turn to exploit the symmetry inherent in regular \(n\)-gon linkages. Recall that given three mutually skew spatial lines there exists a unique doubly ruled quadric, passing through them which is (generically) a hyperboloid of one sheet.  One of its two ruling families consists exactly of the lines that intersect all three; this ruling family is called the \emph{regulus} determined by the triple. These two reguli (and the associated hyperboloid) arise naturally in our construction.\\

\noindent We can further claim that:

\begin{theorem} 
 A generic regular $n$-gon linkage  is a hypo-paradoxical linkage.  
 \label{thm}
\end{theorem}

\begin{proof}

The linkage described in Figure~\ref{f:regular_ngon} can be understood as a physical realization of a segment of a regulus. This surface arises naturally when a cylindrical bundle of strings is twisted (see Figure \ref{f:Hyberboic}). The initially vertical generatrices become skew lines lying on a one-sheeted hyperboloid ruled surface $\frac{x^2}{a^2} + \frac{y^2}{a^2} - \frac{z^2}{c^2} = 1$, 
for some \( a, c \in \mathbb{R}^+ \).  The proof is four folded:

\begin{figure}[htb]
\centering
\scalebox{0.4}{\includegraphics{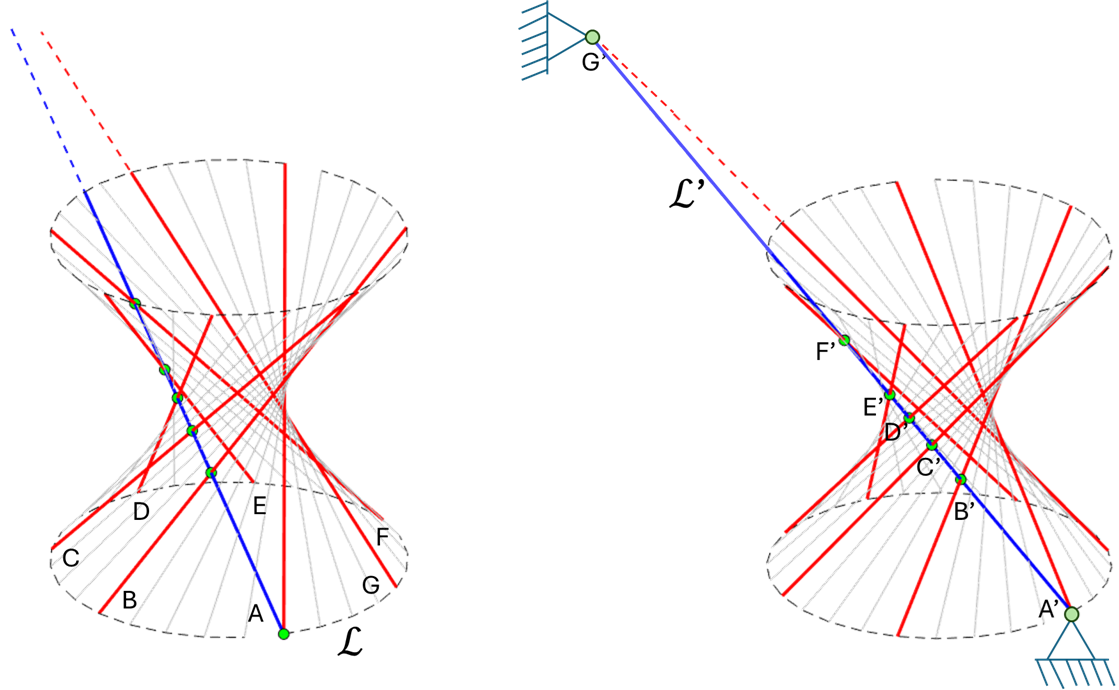}}
\caption{Two projections of a seven ruling lines regulus depicted as red lines. These lines then, intersects sequentially (green dots) a blue line from the conjugate set of lines (In fact, they intersect all lines in the conjugate set of lines). The regular $7$-gon linkage $\LL$ with joints $ABCDEFG$ is located at the lower base of the regulus, with the red lines representing its joints. the blue aligned 7-bar $\LL'$  is realized on the conjugate blue line $A'B'C'D'E'F'G'$.}
\label{f:Hyberboic}
\end{figure}

\begin{enumerate}
    
\item
Observation \ref{lemma:axis} states that from a kinematic perspective, the essential object of study is not the absolute position of each rigid link but the set of \emph{relative transformations} — i.e., the relative positions and orientations between all coordinate frames rigidly attached to the joints. Therefore, the specific embedding of the linkage in \(\mathbb{R}^3\) is irrelevant to its mobility. The only thing that matters is the screws' positions and orientations. 

{\color{black}To see why this holds, consider a linkage composed of rigid links and fixed shafts (the joint axes).
At each end of a link, imagine a cylindrical sleeve that rotates about the corresponding shaft.
Now conceptually dilate these sleeves indefinitely along the direction of each axis.
In doing so, imagine every rigid link transformed into a (rigid) ruled surface whose generators are  the corresponding pair of joint axes. Obviously, the “dilated” linkage behaves exactly like the original one, (disregarding possible self-collisions between them).

From this dilated linkage  one can realize an infinite family of “equivalent” linkages by selecting an arbitrary continuous curve on each surface, stretching from one axis to the other.}

\item
In particular, since the joint axes of the regular \( n \)-gon $\LL$ lie on a regulus, and since any such regulus intersects the second ruling family of the regulus (called the \emph{conjugate regulus}) \cite{gallucci1906studio} , one can choose a specific line \( \LL' \) from the second family such that all joint axes intersect \( \LL' \) see Figure \ref{f:Hyberboic}.

\item
Moreover, the structure of the hyperboloid reguli ensures that the intersections $\xi_i \cap \LL'$ of the axes \( \xi_1, \dots, \xi_n \) of the linkage $\LL$   with the line \( \LL' \) occur in a sequential monotonic order along \( \LL' \) (i.e., the axis \(\xi_1\) intersects \(\LL\) first, \(\xi_2\) second, and so on).  {\color{black} Explicitly, consider the canonical hyperboloid $x^2+y^2-z^2=1$.  Set $n$ equally spaced points $\C_j=\big(\cos\frac{2\pi j}{n},\sin\frac{2\pi j}{n},0\big)$ on the middle circle   on its waist.  The ruling line
$\LL_{k}(t)$ is given by  $\big(\cos\frac{2\pi k}{n},\sin\frac{2\pi k}{n},0\big)+ \ t\,\big(\sin\frac{2\pi k}{n},-\cos\frac{2\pi k}{n},1\big)$.

Each point $\C_j$ is mapped  to the intersection of the opposite ruling through that point with $\LL_k(t)$ at:
\begin{equation}
 t_j=  -\tan\!\Big(\tfrac{\pi (j-k)}{n}\Big) 
\end{equation}
which are ordered along the extended line $\LL_{\alpha}(\mathbb{R})\cup\{\infty\}$ in the same cyclic order as the $n$ points on $\C$ (possibly reversed).  Finally, 
since all reguli $x^\top Ax=1$ are orientation-preservingly \emph{diffeomorphic} via the linear map $x\mapsto Mx$, for $M\in GL(n,\R)$,  this concludes the claim (for a deeper discussion on quadrics and order preserving transformations see \cite{harris2013algebraic}).

}
% matlab: regulus_to_circle_demo.m

In the non-generic cases which we exclude here, this equivalence fails to occur:

  \begin{enumerate}[label=(\roman*),leftmargin=*]
    \item The case where no line $\LL'$ intersects all axes. that is, in the case  $\bxi_1 = \hat{\bm{k}}$ 
    \item In cases where one or more lines intersect all the axes but the order of intersection is not preserved, this occurs when the screw $\{\xi_i\}_{i=1}^n$ form a pencil of lines.
\end{enumerate}
\item
A screw ordering (say $A\to B \to C \to  \dots$, as in Figure \ref{f:regular_ngon}) implies that the linkage can be equivalently realized, without changing the relative geometry of its axes, as a closed chain where all the joints are aligned along a common direction and placed consecutively $A'\to B' \to C' \to  \dots$ along \( \LL' \).
 The preservation of the order of intersection, ensures therefore, that the linkage $\LL'$ is not folded on itself but is globally aligned and thus the configuration is \emph{kinematically locked}. 
 
{\color{black} As depicted in Figure \ref{f:algined}, the corresponding straight bar-on-line realization of the linkage possesses an isolated solution to the closure constraints.}
 \end{enumerate}
That is, no motion is possible despite the system appearing to have positive mobility under naive application of the Chebyshev–Gr\"ubler–Kutzbach criterion - the linkage is \emph{hypo-paradoxical}.
 \begin{figure}[h]
\centering
\scalebox{0.5}{\includegraphics{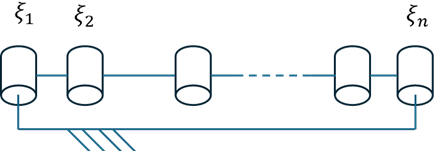}}
\caption{An aligned configuration of a linkage is kinematically locked.}
\label{f:algined}
\end{figure}

\end{proof}

A supplementary video is provided, illustrating how the regular 7-gon mechanism becomes locked, see \cite{Shvalb2025}.\\
 This can be stated as:

\begin{proposition}
\label{thm:hypo}
A closed linkage $\LL$ having a set of screws $\{\xi_i\}_{i\in\II}$ all intersecting a line in a monotone manner is hypo-paradoxical.
\end{proposition}

Up to rigid motion, every regular $n$-gon linkage admits a unique waist realization in which the axes $\{\xi_i\}_{i=1}^n$ have no radial component relative to the symmetry axis of the associated regulus. This realization is parametrized by a single elevation angle $\theta\in(0,\pi)$.
 
\label{waist}

%
%%%%%%%%%%%%%%%%%%%%%%%
\section{From Aligned Axes to Hypo-Paradoxical Linkages}
%%%%%%%%%%%%%%%%%%%%%%%

Carefully considering Proposition \ref{thm:hypo} it is evident that the family of hypo-paradoxical linkages can be trivially extended by considering other ruled surfaces such as \emph{Elliptic regulus, Parabolic Conoid}  and so on. This implies a much larger linkage set of hypo-paradoxical mechanisms:

\begin{figure}[h]
\centering
\scalebox{0.5}{\includegraphics{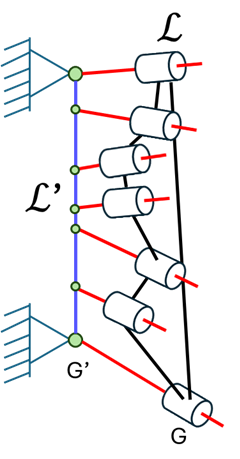}}
\caption{A generic hypo-paradoxical linkage $\LL$ generated reversly from $\LL'$.}
\label{f:general_hypo}
\end{figure}

Applying an opposite perspective can be even more productive  (see Figure \ref{f:general_hypo}): (i) Fix a spatial line $\LL'$; (ii) Specify a collection of screw axes   $\{\bxi_i\}_{i=1}^{n}$ such that each \(\bxi_i\) intersects (or is tangent to) \(\LL'\); (iii) This setting can therefore be thought of as a positively aligned linkage and therefore, kinematically locked (rigid); (iv) For the moment, regard the axes \(\{\bxi_i\}\) as abstract kinematic data, without committing to any particular link geometry; (v) Arbitrarily select a point on each axis and connect these points sequentially into a spatial linkage. These attachment points should be chosen in the order induced by their appearance along the line $\LL'$. Thus, we obtained a generic hypo-paradoxical closed kinematic chain $\LL'$.

In this context, if the mapping of the joints of $\LL$  to those of $\LL'$, are cyclically shifted once, the linkage is also rigid as exemplified in Figure \ref{f:paradoxical_linkage_1}.
Intuitively, since joints $\{A_i\}_2^n$ are bound to rotate perpendicularly to the first axis the linkage is hypo-paradoxical. 
In the context of the discussion above, notice that the linkage is hypo-paradoxical  since the triplet of lengths of link $A_1'A_2'$ in $\LL'$, \ the distance $A_1'A_n'$ and the sum of the links' lengths  $\sum_{i=2}^{n-1} A_i'A_{i+1}'$ satisfy the triangle inequality only in this configuration.

\begin{figure}[h]
\centering
\scalebox{0.5}{\includegraphics{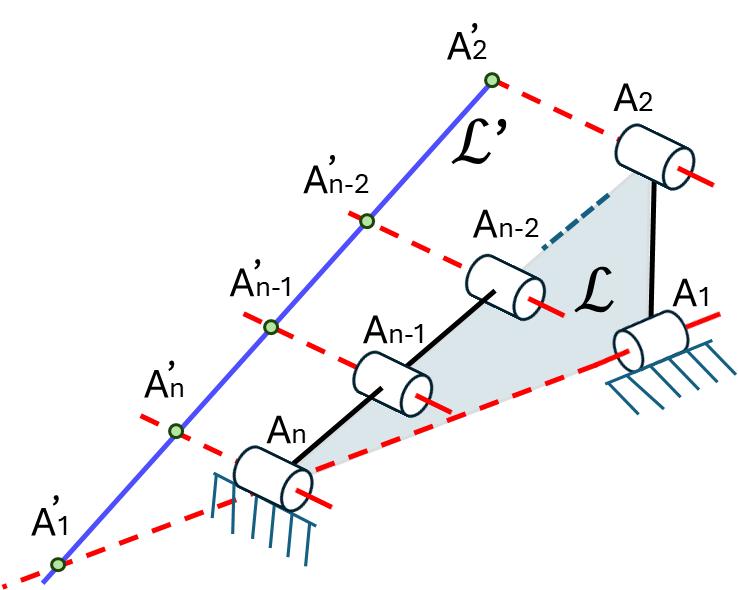}}
\caption{An hypo-paradoxical $n$-bar linkage. The intersection of all screws with $\LL'$ is mapped via a single cyclic  shift.}
\label{f:paradoxical_linkage_1}
\end{figure}

Note, however, that in cases where the mapping from $\LL$ to $\LL'$ is not ordered, the linkage is not expected to be hypo-paradoxical as further explained in Section \ref{sec:bennet}. 
%

%%%%%%%%%%%%%%%%%%%%%%%
\section{Immobility Margin}
%%%%%%%%%%%%%%%%%%%%%%%

The geometry in which a linkage is perfectly aligned along a line \(\LL'\) occupies a set of \emph{measure zero} in its \emph{moduli space} (the space of all design parameters). Equivalently, the subset of realizations \(\LL\) that satisfy the alignment conditions described above also has measure zero.
This follows because:
\begin{enumerate}
    \item The moduli space of \( n \) spatial lines has dimension \( 4n \). Each line is represented by Pl\"ucker coordinates \(\bxi_i = (\mathbf{d}_i, \mathbf{m}_i) \in \mathbb{P}^5\), subject to the constraint \(\mathbf{d}_i \cdot \mathbf{m}_i = 0\). This reduces the six-dimensional ambient space to 4 degrees-of-freedom per line.

    \item Requiring that all \( n \) lines intersect a common line \( \bxi \) imposes \( n \) scalar constraints, each given by the bilinear Pl\"ucker incidence relation:
    \[
    \langle \bxi_i, \bxi \rangle = \mathbf{d}_i \cdot \mathbf{m} + \mathbf{d} \cdot \mathbf{m}_i = 0.
    \]
    However, since the line \( \ell \) is not fixed and itself has $4$ degrees of freedom, only \( n - 4 \) of these constraints are independent.
\end{enumerate}

Thus, the subset of configurations where all lines intersect a common line lies in a subvariety of codimension \( n - 4 \) in the \( 4n \)-dimensional moduli space. For \( n > 4 \), this codimension is positive, implying that such configurations form a set of measure zero in the generic case. In practice, this means that perfect alignment is highly non-generic: any small perturbation due to fabrication error can lead to a loss of alignment. Nevertheless, being close to alignment, is expected to result with a limited work space movement.\\

To capture how close a configuration is to this degenerate, immobile state, we define a measure of proximity of a linkage  to immobility:
\begin{definition} 
The \emph{immobility margin} $\IM$ of a linkage $\LL$ is defined as the infimum of the total distance between corresponding joints \( A_1, A_2, \dots,A_n \) of $\LL$ and the joints \( A_1', A_2', \dots, A_n \) of an aligned, equivalent linkage \( \LL'_i \), taken over all all-intersecting lines \( \{\LL'_i\}_{i \in \II} \).  Formally:
\begin{equation}
    \IM_\LL = \inf_{i \in \II} \left\{ \max_{j-1}^n \operatorname{dist}(A_j, A'_j) \right\}.
\end{equation}
\end{definition}
For example in Figure~\ref{f:general_hypo} and in Figure~\ref{f:margin} the immobility margin is $dist(G, G')$ the farthest  distance between the corresponding joints in $\LL$  and those in $\LL'$.\\
\begin{figure}[h]
\centering
\scalebox{.47}{\includegraphics{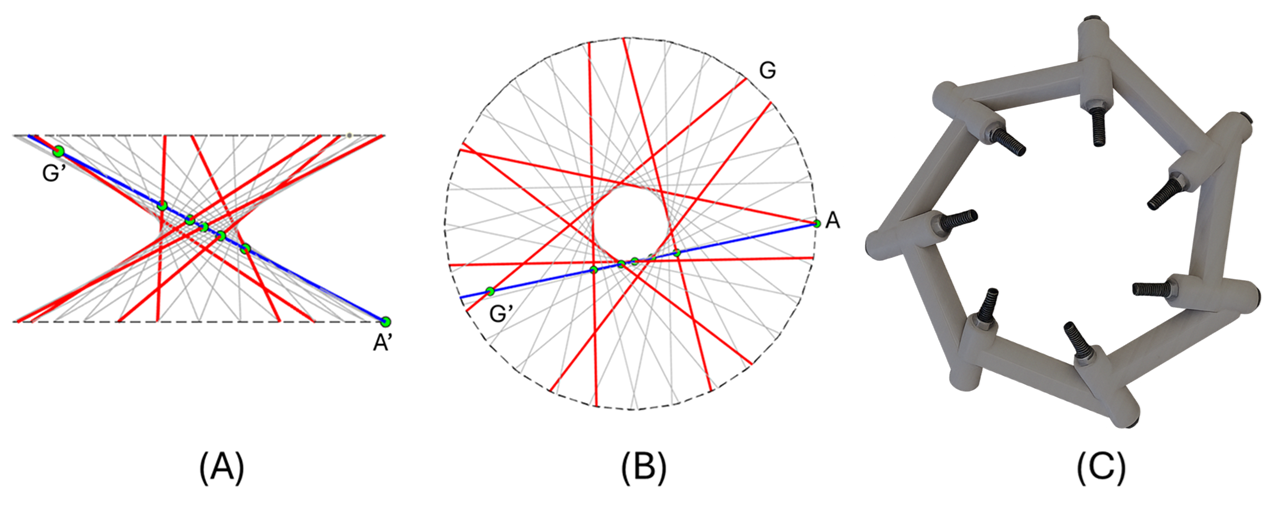}}
\caption{(A) and (B) A regular $7$-gon whose distance between $G$ and $G'$ smaller than that illustrated in Figure~\ref{f:Hyberboic}; (C) A 3D printer realization of a hypo-paradoxical regular $7$-gon with $\alpha=80^\circ$ which admits a workspace diameter of $\approx 0.07$ link-length.}
\label{f:margin}
\end{figure}
Intuitively, when the original linkage is geometrically "far from" its corresponding aligned configuration, even minor fabrication errors in the "original" linkage $\LL$ will prevent the equivalent linkage from achieving full alignment, making the immobility nature more sensitive to fabrication errors. 

To further examine the immobility margin reasoning, consider a general linkage as in Figure \ref{f:general_hypo} denoted by $\mathcal{L}$. Its corresponding aligned version is marked by $\mathcal{L}'$, as described above. Next, we assume that the joint screws of $\mathcal{L}$ are perturbed slightly so that they no longer lie on the common line; We denote this slightly non-aligned linkage   $\LL''$.  Fix two sets coordinates systems $o''_i$ and $o_i$ on each of the screws at $\LL''$ and at $\LL$ and bear in mind that each screw rotates about its neighboring counterparts. Thus, one can deduce the workspace (i.e., the maximal travel of $o_m$ for a certain $m$). knowing the travel of $o''_m$ together with the rotation of its  screw $\hat{\bxi}_m$. 

Let \( \{ e^{\hat{\bxi}_i \theta_i}\}_{i=1}^{n}   \) be the homogeneous transformation matrices representing successive joint motions about screws $\bxi_i= [\bomega_i,\bq_i \times \bomega_i]^\top$.  The middle joint $\hat{\bxi}_m$ at $\LL''$   undergoes translation. To bound the maximal translation possible consider an extreme case in which all screws are re-positioned onto a common plane, each at a constant distance from a reference line, in a zigzag pattern that alternates "above" and "below" the line. This configuration maximizes the effective linkage length and in the extreme case where all screw are parallel, results in the largest possible lateral displacement of the middle joint from the alignment line (i.e., and stretching the "planar" linkage by holding the middle vertex until it forms a triangle).

The \emph{upper bound} for the screw translation is therefore:
\begin{equation}
 D_{\LL''} =
\begin{cases}
2(n-1)M_{\LL''}\Delta\alpha, & \text{if } n \text{ is odd}, \\
2(n-2)M_{\LL''}\Delta\alpha, & \text{if } n \text{ is even}.
\end{cases}   
\label{eq:M}
\end{equation}

where the factor of $2$ reflects the reachability of the opposing configuration by symmetry.

\begin{remark}(Screw rotation) Consider linkage  $\LL''$ and its configuration space  $\{\theta_i \}_{i=1}^n$. There is a small \( \theta_0 \) which is the  maximal angular rotation travel possible \( \theta_i\le\theta_0, \forall i \) . The net angular rotation of  screw $\bxi_m$ under the map $T_m=\prod_{i=0}^{m-1}  e^{\hat{\bxi}_i \theta_i}$  is bounded by a $(m \theta_0)$-cone, i.e., $\angle(\bomega_m,  T_m \bomega_m) \leq m \theta_0$. Nevertheless, since $\LL''$ linkage is nearly aligned $\theta_0$ is very small (not infinitesimal). Therefore, we may assume that the screw translation is dominant over the screw rotation. 
\end{remark}

To quantify the relationship between immobility margin and actual workspace, we simulated the maximal diameter of the configuration space for regular $7$-gon linkages with inclination angles $\alpha = \{20^\circ,25^\circ,\cdots, 70^\circ\}$. For each $\alpha$, we applied a \emph{swarm particle} optimization algorithm that assumed small fabrication errors $\|\vec{\Delta\alpha}\|=0.5^\circ$, modeled as perturbations in the joint directions.  In Fig.~\ref{f:hist} the empirically observed maximal workspace diameters (diamond markers) are plotted versus \(\alpha\). The dashed curve is the theoretical bound for the workspace diameter \(D_{\LL''}\), calculated from the immobility margin according to Equation \ref{eq:M}. For each \(\alpha\) we ran $150$ Monte Carlo simulations, sampling random perturbation directions \(\vec{\Delta\alpha}\) of the same size; the distribution of the resulting workspace diameters is shown as two dimensional histograms beneath the simulated values.

\begin{figure}[h]
 \centering
\scalebox{0.5}{\includegraphics{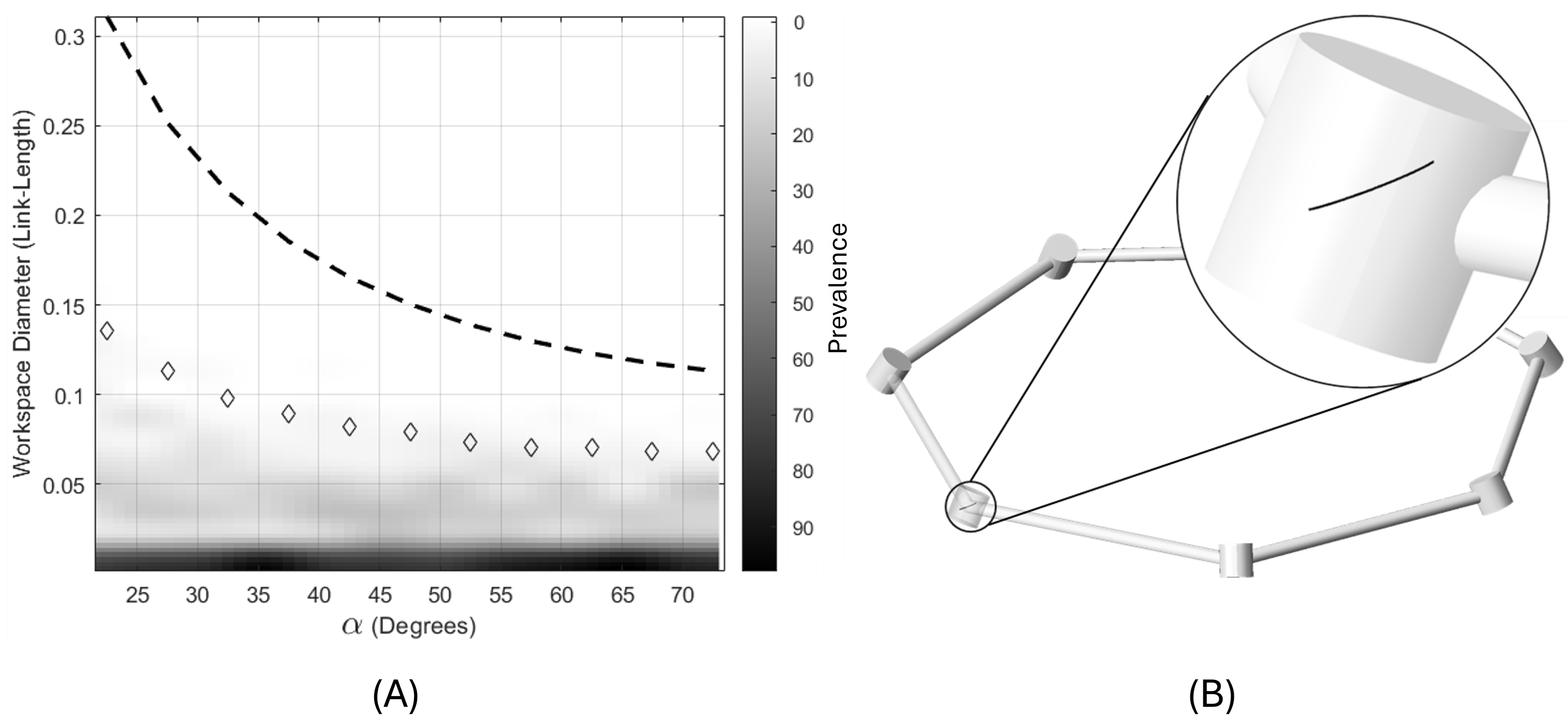}}
\caption{(A) The simulated experiment for different regular $7$-gons, parametrized by the elevation angle $\alpha$ and assuming $\|\vec{\Delta\alpha}\|=0.5^\circ$. The dashed line is the upper bound computed according to $12M_{\LL''}\Delta\alpha$ as in Equation \ref{eq:M}. The diamond markers indicate the maximal workspace diameters computed applying particle swarm optimization over $\vec{\Delta\alpha}$. The shaded part at the bottom correspond to a random set of $\vec{\Delta\alpha}$, i,e., a vertical portion of the figure $\alpha=\alpha_0$ can be viewed as an histogram  for a $7$-gon with $\alpha_0$ of a set of random fabrication errors $\vec{\Delta\alpha}$. The black shade at the bottom of the figure indicate that in most cases the workspace diameter of a fabricated $7$-gon is  small. (B) A simulated example for $\alpha=20^\circ$ and the resultant trajectory of middle joint (the small solid curve) when an error $\|\vec{\Delta\alpha}\|=0.5^\circ$ in the joint directions, is incorporated. }
\label{f:hist}
\end{figure}

%%%%%%%%%%%%%%%%%%%%%%%
\section{The Bennett Linkage Explained: Simplicity Behind the Paradox}
\label{sec:bennet}
%%%%%%%%%%%%%%%%%%%%%%%

The Bennett linkage is a well-known paradoxical mechanism which has one degree of freedom, despite the classical Chebyshev–Gr\"ubler–Kutzbach formula predicting $M=-2$. A correction to the formula by incorporating the order of the linkage's finite symmetry group \cite{hayes2025chebyshev} yields the correct mobility count. Nevertheless, although analytically powerful, the approach does not provide an intuitive geometric account of the Bennett mechanism’s mobility.

Following the same reasoning as above, the one-dimensional set of the Bennett linkage's joints axes trace a family of hyperboloid of one sheet surfaces,   i.e., the four revolute axes of the linkage lie along rulings of this surfaces (see Figure \ref{f:sheets}).  Explicit algebraic relations linking the parameters of the Bennett linkage and the configuration angle to the associated   hyperboloid are derived in~\cite{baker1988} (Baker calls them J-hyperboloids and further defines the L-hyperboloids which are defined by the links rather than the joints).\\

Consider the linkage at its maximally convex configuration $c_0$ as it is classically presented, that is, the angles between consecutive revolute axes attains their maximum. The linkage at $c_0$ defines a hyperboloid (say, in a canonical representation, having no coupling between $xz$ and $yz$). As the linkage evolves through its kinematic motion, the associated hyperboloid changes continuously from configuration to configuration  - its orientation and embedding in space vary but it consistently retains a hyperboloid of one sheet shape:
\[
\mathbf{x} ^\top Q \mathbf{x} = 0,
\]
where \( \mathbf{x} = (x,y,z, 1) ^\top \) are written in homogeneous coordinates
and \( Q \) is a symmetric \( 4 \times 4 \) matrix. The classification of these surfaces relies on the eigenvalue signature of the \( 3 \times 3 \) submatrix corresponding to the quadratic part. A surface is identified as a hyperboloid of one sheet if this submatrix has exactly two eigenvalues of one sign and one of the opposite sign, i.e., signature \( (2,1) \) or \( (1,2) \) (other quadric types, such as ellipsoids, paraboloids, or hyperboloids of two sheets, correspond to different combinations of positive, negative, and zero eigenvalues). The persistence of this signature across all configurations indicates that the hyperboloid of one sheet structure is preserved throughout.
\begin{figure}[htb]
\centering
\scalebox{.5}{\includegraphics{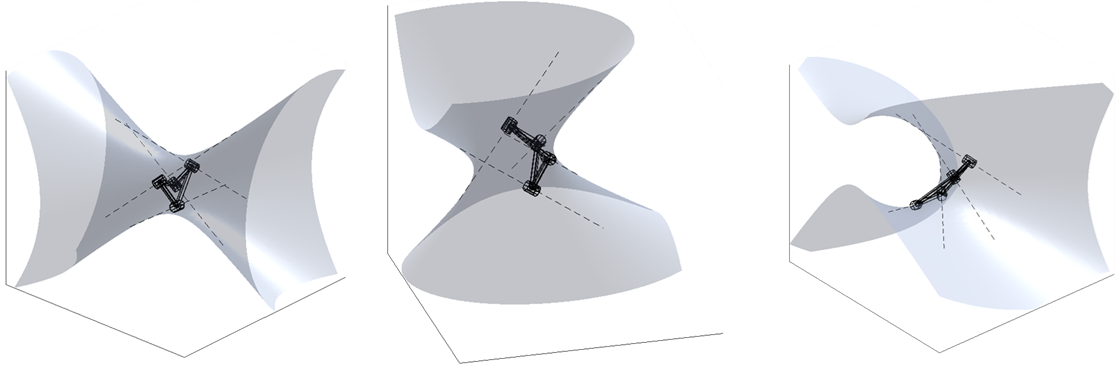}}
\caption{Three configurations of a given Bennett linkage with $\theta=0^\circ,40^\circ, 100^\circ$, each shown with the corresponding hyperboloid of one sheet whose rulings contain the four revolute axes.}
\label{f:sheets}
\end{figure}

Since each hyperboloid of one sheet admits two distinct families of rulings, it follows, as before, that there exists a one-parameter infinite family of lines (the conjugate set), that intersect all four axes at every configuration. In particular, for each step along the motion, one can select a specific ruling $\LL'$ that intersects all four revolute axes. 

\begin{lemma}
Let \( \mathcal{L} \) be a Bennett linkage, and let \( \mathcal{L}' \) denote its aligned linkage, i.e., the spatial line where all revolute axes lie on a single ruling of the associated regulus. Then the associated linkage of \( \mathcal{L}' \) is aligned but it is also folded, meaning the intersection points on $\LL'$ are not ordered in a monotone manner.
\end{lemma}

\begin{proof}
Consider Bennett linkage $\LL$ at its maximally convex configuration \( \LL(c_0) \). In this case it is easy to check that the geometric order induced by the ruling necessarily disagrees with the linkage connectivity. In other words,  the links of the aligned linkage \( \mathcal{L}'(c_0) \)  do not appear in a consistent sequential arrangement along the ruling direction.
\begin{figure}[htb]
\centering
\scalebox{.4}{\includegraphics{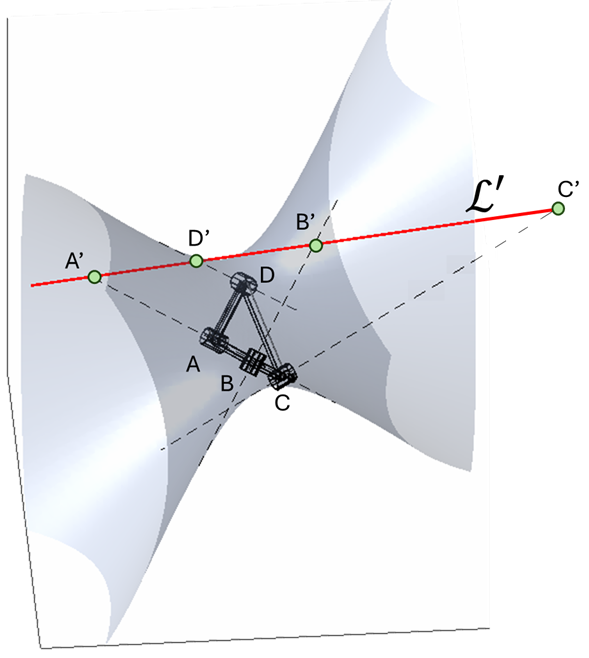}}
\caption{A Bennett mechanism in a general configuration. All joint axes $\xi_i$ lie on a common regulus, while the line $\LL'$ lies on the conjugate regulus. The intersections of the axes $\xi_i$ with the line $\LL'$ are  non-ordered (i.e., do not satisfy the conditions of proposition \ref{thm:hypo}).}
\label{f:bennet_non}
\end{figure}
Assume for contradiction that there exists a configuration \( c_2 \in \mathcal{C}(\mathcal{L}) \) such that \( \mathcal{L}'(c_2) \) is positively aligned, i.e., all links point in the same direction along the ruling and appear in the correct sequence. Since the configuration space \( \mathcal{C}(\mathcal{L}) \cong \mathbb{S}^1 \) is path-connected, so is the configuration space \( \mathcal{C}(\mathcal{L}') \). Therefore, there exists a continuous path of configurations from \( c_0 \) to \( c_2 \), and hence an intermediate configuration \( c_1 \) for which at least two axes of \( \mathcal{L}'(c_1) \) coincide. However, this contradicts the geometric constraint that the axes of \( \mathcal{L}' \) lie on a single ruling of the regulus, since lines in a single rulings do not intersect. Therefore, no such ordered configuration \( c_2 \) exists, and \( \mathcal{L}' \) must remain folded in all configurations.
\end{proof}

In light of this lemma, although the Bennett linkage is aligned along a ruling of a quadric surface, it is effectively folded onto itself and therefore does not satisfy the conditions of Proposition~\ref{thm:hypo}. Recalling that any three screw lines lying on a regulus define a \emph{screw system}~\cite{dandurand1984rigidit6,huang2012theory}, a fourth screw line located on the same regulus is linearly dependent on the former three. In other words, the joint screws of the Bennett linkage remain dependent throughout its entire kinematic cycle, thereby ensuring its mobility.

\section*{Conclusion}
This paper introduced the notion of \emph{hypo-paradoxical linkages}: mechanisms that defy the classical mobility count (Chebyshev–Gr\"ubler–Kutzbach) by being kinematically rigid even though the count predicts $M>0$. We showed that such rigidity can arise from subtle geometric alignment. In particular, when all joint screws intersect a common line in a monotone (order‑preserving) fashion, collapsing the available motion. Using analytic and geometric tools, we derived general conditions for this behavior and illustrated them on the family of regular $n$-gon linkages. 

We further applied the same geometric viewpoint to the classical Bennett linkage, explaining why, despite its paradoxical reputation. 

Future work will investigate \emph{hyper-paradoxical mechanisms} that incorporate prismatic joints, exploring how mixed revolute–prismatic alignments can either induce or relieve mobility deficits.


\clearpage


 \begin{thebibliography}{99}
 

\bibitem[Schicho(2020)]{schicho2020}
Schicho, J. (2020).
\newblock And Yet It Moves: Paradoxically Moving Linkages in Kinematics.
\newblock {\em arXiv preprint}, arXiv:2004.12635.

\bibitem[Li et~al.(2018)Li, Schicho, and Schr\"ocker]{li2018}
Li, Z., Schicho, J., \& Schr\"ocker, H.-P. (2018).
\newblock A Survey on the Theory of Bonds.
\newblock {\em arXiv preprint}, arXiv:1808.03138.

\bibitem[Guerreiro et~al.(2021)Guerreiro, Li, and Schicho]{guerreiro2021}
Guerreiro, T. D., Li, Z., \& Schicho, J. (2021).
\newblock Classification of Higher Mobility Closed-Loop Linkages.
\newblock {\em arXiv preprint}, arXiv:2103.04799.

\bibitem[Bennett(1903)]{bennett1903}
Bennett, G. T. (1903).
\newblock A New Mechanism.
\newblock {\em Engineering}, 76:777--778.
\newblock Reprinted in \textit{Engineering} (London), 1903.

\bibitem[Bricard(1897a)]{bricard1897}
Bricard, R. (1897a).
\newblock M\'emoire sur la th\'eorie de l'octa\'edre articul\'e.
\newblock {\em Journal de Math\'ematiques Pures et Appliqu\'ees}, 3:113--148.
\newblock Translated as ``Memoir on the Theory of the Articulated Octahedron'' by E.~A. Coutsias, 2010.

\bibitem[Husty and Schr\"ocker(2012)]{husty2012}
Husty, M. L., \& Schr\"ocker, H.-P. (2012).
\newblock Kinematics and Algebraic Geometry.
\newblock In J. M. McCarthy (Ed.), {\em 21st Century Kinematics: The 2012 NSF Workshop}, pp. 85--123. Springer.

\bibitem[Sarrus(1853)]{sarrus1853}
Sarrus, P.-F. (1853).
\newblock Note sur la construction d’un mouvement rectiligne.
\newblock {\em Comptes Rendus Hebdomadaires des S\'eances de l'Acad\'emie des Sciences}, 36:1036--1038.

\bibitem[Cox et~al.(2007)]{CoxLittleOShea}
Cox, D., Little, J., \& O'Shea, D. (2007).
\newblock {\em Ideals, Varieties, and Algorithms} (3rd ed.).
\newblock Springer.

\bibitem[Gallucci(1906)]{gallucci1906studio}
Gallucci, G. (1906).
\newblock {\em Studio della figura delle otto rette e sue applicazioni alla geometria del tetraedro ed alla teoria delle configurazioni}.

\bibitem[Baker(1988a)]{baker1988}
Baker, J. E. (1988a).
\newblock The Bennett linkage and its associated quadric surfaces.
\newblock {\em Mechanism and Machine Theory}, 23(2):147--156.

\bibitem[Lee and Herv\'e(2011)]{leea2011synthesize}
Lee, C.-C., \& Herv\'e, J. M. (2011).
\newblock Synthesize new 5-bar paradoxical chains via the oblique circular cylinder.
\newblock {\em Mechanism and Machine Theory}, 46:784--793.

\bibitem[Bricard(1897b)]{bricard1897octaedre}
Bricard, R. (1897b).
\newblock M\'emoire sur la th\'eorie de l'octa\`edre articul\'e.
\newblock {\em Journal de l'\'Ecole Polytechnique}, 2\textsuperscript{e} s\'erie, XI(71):113--148.

\bibitem[Delassus(1922)]{delassus1922}
Delassus, E. (1922).
\newblock Les cha\^ines articul\'ees ferm\'ees et d\'eformables \`a quatre membres.
\newblock {\em Bulletin des Sciences Math\'ematiques (2\textsuperscript{e} s\'erie)}, 46:283--304.

\bibitem[Myard(1931)]{myard1931}
Myard, F. E. (1931).
\newblock Contribution \`a la g\'eom\'etrie des syst\`emes articul\'es.
\newblock {\em Bulletin de la Soci\'et\'e Math\'ematique de France}, 59:183--210.

\bibitem[Goldberg(1943)]{goldberg1943}
Goldberg, M. (1943).
\newblock New five-bar and six-bar linkages in three dimensions.
\newblock {\em Transactions of the ASME}, 65:649--663.

\bibitem[Waldron(1967)]{waldron1967}
Waldron, K. J. (1967).
\newblock A family of overconstrained linkages.
\newblock {\em Journal of Mechanisms}, 2:201--211.

\bibitem[Waldron(1979)]{waldron1979}
Waldron, K. J. (1979).
\newblock Overconstrained linkages.
\newblock {\em Environment and Planning B: Planning and Design}, 6:393--402.

\bibitem[Baker(1979)]{baker1979}
Baker, J. E. (1979).
\newblock The Bennett, Goldberg and Myard linkages -- in perspective.
\newblock {\em Mechanism and Machine Theory}, 14(4):239--253.

\bibitem[Wohlhart(1987)]{wohlhart1987}
Wohlhart, K. (1987).
\newblock A new 6R space mechanism.
\newblock In {\em Proc. 7th World Congress on the Theory of Machines and Mechanisms (IFToMM)}, Sevilla, Spain, vol. 1, pp. 193--198.

\bibitem[Wohlhart(1991)]{wohlhart1991}
Wohlhart, K. (1991).
\newblock Merging two general Goldberg 5R linkages to obtain a new 6R space mechanism.
\newblock {\em Mechanism and Machine Theory}, 26(2):659--668.

\bibitem[Dandurand(1984)]{dandurand1984rigidit6}
Dandurand, A. (1984).
\newblock Rigidit\'e des r\'eseaux spatiaux compos\'es.

\bibitem[Huang et~al.(2012)Huang, Li, and Ding]{huang2012theory}
Huang, Z., Li, Q., \& Ding, H. (2012).
\newblock {\em Theory of Parallel Mechanisms}.
\newblock Springer.

\bibitem[Connelly(2005)]{connelly2005generic}
Connelly, R. (2005).
\newblock Generic global rigidity.
\newblock {\em Discrete \& Computational Geometry}, 33(4):549--563.

\bibitem[Kapovich and Millson(2002)]{kapovich2002moduli}
Kapovich, M., \& Millson, J. J. (2002).
\newblock Moduli spaces of linkages: a survey.
\newblock {\em Canadian Mathematical Bulletin}, 45(4):643--651.

\bibitem[Hayes and Colla(2025)]{hayes2025chebyshev}
Hayes, M. J. D., \& Colla, A. (2025).
\newblock The Chebyshev--Gr\"ubler--Kutzbach Mobility Criterion Revisited.
\newblock In {\em CCToMM Symposium on Mechanisms, Machines, and Mechatronics}, pp. 28--39. Springer.

\bibitem[Veblen and Young(1918)]{veblen1918projective}
Veblen, O., \& Young, J. W. (1918).
\newblock {\em Projective Geometry}, vol. 1.
\newblock Ginn.

\bibitem[Medina et~al.(2013)Medina, Taitz, Ben Moshe, and Shvalb]{medina2013c}
Medina, O., Taitz, A., Ben Moshe, B., \& Shvalb, N. (2013).
\newblock C-space compression for robots motion planning.
\newblock {\em International Journal of Advanced Robotic Systems}, 10(1):6.

\bibitem[Medina et~al.(2015)Medina, Shapiro, and Shvalb]{medina2015motion}
Medina, O., Shapiro, A., \& Shvalb, N. (2015).
\newblock Motion planning for an actuated flexible polyhedron manifold.
\newblock {\em Advanced Robotics}, 29(18):1195--1203.

\bibitem[Medina et~al.(2016a)Medina, Shapiro, and Shvalb]{medina2016minimal}
Medina, O., Shapiro, A., \& Shvalb, N. (2016a).
\newblock Minimal actuation for a flat actuated flexible manifold.
\newblock {\em IEEE Transactions on Robotics}, 32(3):698--706.

\bibitem[Blanc and Shvalb(2023)]{blanc2023symmetric}
Blanc, D., \& Shvalb, N. (2023).
\newblock Symmetric configuration spaces of linkages.
\newblock {\em Journal of Applied and Computational Topology}, 7(3):527--569.

\bibitem[Medina et~al.(2016b)Medina, Shapiro, and Shvalb]{medina2016kinematics}
Medina, O., Shapiro, A., \& Shvalb, N. (2016b).
\newblock Kinematics for an actuated flexible n-manifold.
\newblock {\em Journal of Mechanisms and Robotics}, 8(2):021009.

\bibitem[Shvalb(2025)]{Shvalb2025}
Shvalb, N. (2025).
\newblock Hypoparadoxical linkages.
\newblock YouTube video, Published: September 25, 2025.
\newblock \url{https://www.youtube.com/watch?v=n6-0Zn87cjc}.

\bibitem[Beutelspacher and Rosenbaum(1998)]{beutelspacher1998projective}
Beutelspacher, A., \& Rosenbaum, U. (1998).
\newblock {\em Projective Geometry: From Foundations to Applications}.
\newblock Cambridge University Press.

\bibitem[Beardon(1983)]{beardon1983mobius}
Beardon, A. F. (1983).
\newblock M\"obius Transformations on $\mathbb{R}^n$.
\newblock In {\em The Geometry of Discrete Groups}, pp. 20--55. Springer.

\bibitem[Harris(2013)]{harris2013algebraic}
Harris, J. (2013).
\newblock {\em Algebraic Geometry: A First Course}, vol. 133.
\newblock Springer.

\end{thebibliography}
\end{document}